\documentclass{amsproc}

\usepackage{hyperref}
\usepackage{url,graphicx}
\usepackage{amsthm,amsmath,amssymb,amsfonts}
\usepackage{showexpl}
\usepackage{mathrsfs,xr}
\usepackage[all,arc,curve,color,frame]{xy}
\usepackage{graphics,fullpage}
\xyoption{all}

\newtheorem{theorem}{Theorem}[section]
\newtheorem{definition}[theorem]{Definition}
\newtheorem{example}[theorem]{Example}
\newtheorem{proposition}[theorem]{Proposition}

\theoremstyle{remark}
\newtheorem{remark}[theorem]{Remark}

\numberwithin{equation}{section}

\newcommand{\TR}{\operatorname{Tr}}
\newcommand{\Tr}{\operatorname{tr}}

\newcommand{\g}{\mathfrak{g}} 
\newcommand{\su}{\mathfrak{su}}
\newcommand{\naturalnumber}{\mathbb{N}}

\newcommand{\Aa}{\mathcal{A}}
\newcommand{\Bb}{\mathcal{B}}

\newcommand{\Ss}{\mathcal{S}}

\newcommand{\Mfold}{M}
\newcommand{\Func}{\operatorname{Func}}

\newcommand{\Tt}{\mathcal{T}}

\newcommand{\SO}{\operatorname{s}}
\newcommand{\RA}{\operatorname{r}}

\newcommand{\HOL}{\operatorname{Hol}}

\newcommand{\Cc}{\mathbb{C}}
\newcommand{\Rr}{\mathbb{R}}

\newcommand{\Gg}{\mathcal{G}}
\newcommand{\Qq}{\mathcal{Q}}
\newcommand{\Ggamma}{\mathbf{\Gamma}}

\newcommand{\Plim}{\varprojlim}

\newcommand{\KER}{\operatorname{ker}}
\newcommand{\DIFF}{\operatorname{Diff}}

\newcommand{\EXP}{\operatorname{Exp}}

\begin{document}

\title{Strict Deformation Quantisation of the $G$-connections via Lie Groupoid}

\author{Alan Lai}
\address{Max Planck Institute for Mathematics -- Bonn}
\curraddr{Vivatsgasse 7, 53111, Bonn, Germany}
\email{alan@mpim-bonn.mpg.edu}
\thanks{The research is completed during the author's stay at the MPIM-Bonn}

%
%


\begin{abstract}
Motivated by the compactification process of the space of connections in loop quantum gravity literature.
A description of the space of $G$-connections using the tangent groupoid is given.
As the tangent groupoid parameter is away from zero, the $G$-connections are (strictly) deformation
quantised to noncommuting elements using $C^*$-algebraic formalism.
The approach provides a mean to obtaining a semi-classical limit in loop quantum gravity.
\end{abstract}

\maketitle

\section{Introduction}
In Ashtekar's theory of gravity, a $SU(2)$-connection captures the extrinsic curvature of a space-like leaf
in a time-transversal foliation, and the intrinsic geometry of the leaf is given by a tetrad \cite{ashtekar}.
It is shown that the Einstein-Hilbert functional and Einstein equations can be written
in terms of  $SU(2)$-connections and tetrads, thus
these variables together  recast Einstein's theory of gravity.
By rewriting Einstein's gravity with the connection variables and tetrad variables, 
one could attempt to quantise gravity via the Hamiltonian formalism, and obtain a theory of quantum gravity \cite{ashham}.
The reader may refer to Thiemann's introductory \cite{LQG}.
This article offers an alternative view of the space of connections to ones that appear in other loop quantum gravity literatures, and proposes a semi-classical limit using  strict $C^*$-algebraic  deformation quantisation formalism \cite{rieffel}.

To take the connections as dynamic variables in a quantum theory, one studies wave functions, also known as probability amplitudes, on the space of connections. Unfortunately, the lack of a measure on such a space  poses the first challenge, since in such case probabilities cannot be defined.
One typical solution to obtain a measure on the connection space comes from spaces of progressively 
refined cylindrical functions. In another description, one uses  a finite set of curves to probe the space of 
$G$-connections to obtain a finite dimensional manifold  that depend on the sets of curves. By successively refining 
the finite sets of curves, such as successive triangulation of the manifold, one obtains a pro-manifold that extends
the  original space of connections to the so-called space of generalised connections \cite{baez}.

As a step to quantising gravity in the Ashtekar framework, there are recent developments of describing such an extended space of connections using a spectral triple in noncommutative geometry \cite{AGN,lai}, which captures the geometry
of the space of generalised connections as operators on a Hilbert space.
While the geometries of the base manifold and the  space of $G$-connections  on it are in theory retained, the construction of  the spectral triple is considered too discrete to practically allow one to recapture the geometry of the
base manifold and its $G$-connections.
To rid the discrete description of using finite sets of embedded curves, this article proposes an alternative approach to smoothly probe the space of connections using the tangent groupoid. And then followed by an application of strict deformation quantisation in $C^*$-algebraic formalism, a possible semi-classical limit of can be obtained.
The purposes of this article are to present an alternative idea to the studies of loop quantum gravity
with the goal of obtaining a semi-classical limit, which is what loop quantum gravity still lacks today.

This article contains six small sections, they are arranged as follows:
Section~\ref{sec2}
discusses the traditional way of putting  a measure on the space of connections in loop quantum gravity literature, and the 
short coming of such a method.
In Section~\ref{sec3}, we review the definitions of a Lie groupoid, a Lie algebroid, and tangent groupoid of a Lie groupoid in an elementary way.
Section~\ref{sec4} shows that how one uses the tangent groupoid as a tool 
to model  the space of connections in a smooth way, and discuses the corresponding gauge action.
Section~\ref{sec5} is the first attempt of deforming connections into noncommutating operators acting on a Hilbert space.
Section~\ref{sec6} introduces the notion of strict deformation quantisation in $C^*$-algebra formalism, and an important theorem by Landsman \cite{landsman}, which states that a tangent groupoid defines a strict deformation quantisation. Hence, using this formalism, we can deformation quantise $G$-connections.
Finally,  Section~\ref{sec7} is an outlook that summarises the article and the current state of work toward obtaining a semi-classical limit in loop quantum gravity.

\section{A Measure on the Space of Connections}
\label{sec2}

We start by requiring that $G$ to be a semi-simple, simply connected, compact Lie group, such as $SU(2)$.
And $\g$ to be the corresponding Lie algebra of $G$. $\exp $ is the exponential map from $\g$ to $G$.
The space of smooth $G$-connections $\Aa$ over a manifold $M$ is the space of $\g$-valued $1$-forms over $M$.

Given a smooth manifold $M$ and a smooth (compact) curve $\gamma$  in $M$,
one obtains a map $\HOL_\gamma$ from the space of smooth $G$-connections $\Aa$ to $G$
\[
\HOL_\gamma: \Aa \to G,
\]
given by taking the holonomy 
\[
\HOL_\gamma(A):= \exp\int_\gamma A
\]
of each connection $A\in \Aa$ along the curve $\gamma$.
Here it is assumed that there is a fixed local trivialisation of the principal bundle, so that $\Aa$ is realised as a vector space.

Suppose that there is a smoothly embedded finite graph $\Gamma$ in $M$ (edges are smooth, and the number of them is finite), one can repeat the holonomy evaluation to obtain a map from $\Aa$ to multiple copies of $G$'s, one for each edge.
Thus, one obtains
\begin{eqnarray}
\label{AGN:eqn:coarsegrainedhol}
\HOL_\Gamma:\Aa \to G^{\lvert \Gamma \rvert},
\end{eqnarray}
where $\vert \Gamma \rvert$ denotes the number of edges of the graph $\Gamma$.

\begin{proposition}[\cite{baezsawin,AGN}]
\label{AGN:prop:denserange}  
For any finite
graph $\Gamma$, 
the map
$\HOL_\Gamma$ \eqref{AGN:eqn:coarsegrainedhol} is a surjection.
\end{proposition}

While the space of smooth $G$-connections $\Aa$ lacks structures, such as a measure, it
surjects to $G^{\lvert \Gamma \rvert}$, a compact measure space.
It is done by forgetting the values of the connections outside the edges of the graph. Hence, a lot of information is lost. However, one can imagine that there is a collection of finite graphs with one finer than the other, such that the collection of graphs is dense in the manifold $M$ in a certain sense. Therefore, at every small neighbourhood in the manifold, there is an edge from the collection of graphs in the neighbourhood to probe the holonomies of the connections.
We do not  define the notion of a  {\bf directed system} of finite graphs, but to say that it is defined naturally by the associated groupoid of a directed finite graph \cite{lai}.
From which, there is associated a directed system of compact measure spaces $G^{\lvert \Gamma_j \rvert }$.

To  put the above description in a mathematical context, we state
\begin{theorem}[\cite{AGN,lai}]
Denote by $\overline{\Aa}^\Ggamma$ the projective limit $
\Plim G^{\lvert \Gamma_j \rvert }$.
Suppose that there is a system of of smoothly embedded finite graphs $\Ggamma:=\{\Gamma_j\}_j$, such that the set of vertices is dense in $M$ and every neighbourhood of a point $x\in M$ contains
edges that span the vector space $T_x M$. 
Then there exists an embedding
\[ 
\HOL_\Ggamma: \Aa \hookrightarrow \overline{\Aa}^\Ggamma .
\]
\end{theorem}
From the property of the product (Tychonoff) topology, one has that
\begin{proposition}[\cite{baezsawin,AGN}]
\label{AGN:prop:denserange}  
Let $\Ss=\{\Gamma_i\}_{i\in I}$ be a directed system of finite graphs
in $M$. Then
the image of 
$\Aa$ under $\HOL_\Ggamma$ is dense in  $\overline{\Aa}^\Ggamma$.
That is
\[
\overline{
\HOL_\Gamma(\Aa)}=\overline{\Aa}^\Ggamma.
\]
\end{proposition}

This compactification procedure depends on the system of graphs used in probing the connection space. 
We give some examples of graph systems that provide good compactifications.

\begin{example}
\mbox{}
\begin{enumerate}
\item
\label{AGN:ex:triangulation}
Let $\Tt_1$ be a triangulation of $M$ and $\Gamma_1$ be the graph
consisting of all the edges in this triangulation with
any orientation.
Let $\Tt_{n+1}$ denote the triangulation obtained by
barycentric subdivision of each of the simplices in
 $\Tt_1$ $n$ times. The graph $\Gamma_{n+1}$ is
the graph consisting  the edges of $\Tt_{n+1}$ with consistent
orientation.
 In this way 
$\Ss_\triangle :=\{\Gamma_n \}_{n\in \naturalnumber}$
 is a directed system of finite graphs, and $\Aa$ densely embeds into $\overline{\Aa}^{\Ss_\triangle}$.
%
\item
\label{AGN:ex:dlattice}
 Let $\Gamma_1$ be a finite, 
$d$-dimensional lattice in $M$ and let $\Gamma_2$ denote the
lattice obtained by subdividing each cell in $\Gamma_1$ into $2^d$ cells. 
Correspondingly, let $\Gamma_{n+1}$ denote the lattice obtained by repeating $n$ such subdivisions
of $\Gamma_0$ . In this way 
 $\Ss_\square :=\{\Gamma_n\}_{n\in \naturalnumber}$
 is a directed system of finite graphs, and $\Aa$ densely embeds into $\overline{\Aa}^{\Ss_\square}$.
 \end{enumerate}
\end{example}
Moreover, $\overline{\Aa}^\Ggamma$ is a compact measure space. Therefore, such a procedure of 
surjecting $\Aa$ to a coarse approximation $G ^{\lvert \Gamma_j \rvert }$, then consider the limit of the approximations provide an extension of the space $\Aa$ to a compact measure space $\overline{\Aa}^\Ggamma$.
The space $\overline{\Aa}^\Ggamma$ is called the space of generalized connections.
Now it is possible to consider probability amplitudes on the space of generalized connections, and proceed to canonical quantisation.

The  drawback of this compactification is that the original space of connections $\Aa$ is forever lost in $\overline{\Aa}^\Ggamma$, consequently obtaining a semi-classical limit from  quantisation on $\overline{\Aa}^\Ggamma$  is impossible.
The source of the problem comes from probing the space $\Aa$ with finite graphs, which are very rigid objects that cannot be perturbed.
One can understand or abstract this process of compactifaction of $\Aa$ as probing $\Aa$ with a groupoid, where for the case of a graph $\Gamma$, the groupoid is the associated fundamental groupoid. It is a finitely generated groupoid, which is very discrete and cannot be perturbed.
However, one can replace this discrete groupoid with a smooth groupoid, say a Lie groupoid. This article proposes the use of the tangent groupoid.


\section{Tangent Groupoid}
\label{sec3}

\begin{definition}
\label{groupoid}
A {\bf groupoid} $\Gg$ is a (small) category in which every morphism is invertible.
That is, a set of morphisms $\Gg$ together with a set of objects $X$ such that
\begin{itemize}
\item
there exist surjective structure  maps, called the source and range maps
\[
\Gg  \rightrightarrows ^{\hspace{-0.25cm}^{\RA}}
_{\hspace{-0.25cm}_{\SO}}X ,\]
\item
there exists an injection, called the identity inclusion,
\[
i:X \hookrightarrow \Gg,
\]
\item
there exists a partially-defined associative composition
\[
\Gg\times \Gg \to \Gg
\]
with identities $i(X)$, and
\item
there exists an inversion map
\[
\left( \mbox{ } \right) ^{-1}: \Gg \to \Gg
\] with the usual properties.

\end{itemize}

\end{definition}

\begin{definition}
\label{liegroupoid}
A {\bf Lie groupoid} is a groupoid
$\Gg  \rightrightarrows ^{\hspace{-0.25cm}^{\RA}}
_{\hspace{-0.25cm}_{\SO}}X$ with smooth manifold structures on $\Gg$ and $X$ such that $\SO, \RA$ are 
submersions, the inclusion of $X$ in $\Gg$ as the identity homophism and the composition $\Gg \times \Gg \to \Gg$ are smooth.
\end{definition}

\begin{example}[\cite{tangent}]
\label{liegroupoidex}\mbox{ }
\begin{enumerate}
\item Any Lie group $G$ is a groupoid $G  \rightrightarrows ^{\hspace{-0.25cm}^{\RA}}
_{\hspace{-0.25cm}_{\SO}} \{e\}$ over the identity $e$.
\item
The tangent bundle $T\Mfold$ of a manifold $\Mfold$ forms a Lie groupoid $T\Mfold  \rightrightarrows ^{\hspace{-0.25cm}^{\RA}}
_{\hspace{-0.25cm}_{\SO}}\Mfold$ with the \emph{source} and \emph{range} maps 
$\SO, \RA: T\Mfold\to \Mfold$ given by 
$\SO(x,V_x) = x = \RA(x,V_x)$ for $(x,V_x) \in T\Mfold$, the inclusion $\Mfold \hookrightarrow T\Mfold$ given by
the zero section, and composition $T\Mfold\times T\Mfold \to T\Mfold$ given by $(x,V_x )\times (x,W_x) \mapsto 
(x,V_x + W_x)$.
\item
The product $\Mfold\times \Mfold$ forms a Lie groupoid $\Mfold \times \Mfold
  \rightrightarrows ^{\hspace{-0.25cm}^{\RA}}
_{\hspace{-0.25cm}_{\SO}} \Mfold $ with the \emph{source} and \emph{range} maps 
$\SO, \RA: \Mfold \times \Mfold \to \Mfold$ given by 
$\SO(x,y ) = x $, $ \RA(x,y)=y$, the inclusion $\Mfold \hookrightarrow \Mfold\times \Mfold$ is given by the 
diagonal embedding, and the composition 
$(\Mfold\times \Mfold) \times (\Mfold\times \Mfold) \to \Mfold \times \Mfold$ is given by
$(x,y) \times (y,z) \mapsto (x,z)$.

\item
Given two Lie groupoids 
$\Gg  \rightrightarrows ^{\hspace{-0.25cm}^{\RA}}
_{\hspace{-0.25cm}_{\SO}}X$ and $\Gg'  \rightrightarrows ^{\hspace{-0.25cm}^{\RA'}}
_{\hspace{-0.25cm}_{\SO'}}X'$, their direct product 
$\Gg \times \Gg'  \rightrightarrows ^{\hspace{-0.40cm}^{\RA \times \RA'}}
_{\hspace{-0.40cm}_{\SO \times \SO '}}X\times X'$
is also a Lie groupoid.

\end{enumerate}

\end{example}
\begin{definition}
 A \textbf{Lie algebroid} on a manifold $\Mfold$ is a vector bundle $E$ over $\Mfold$, which  is equipped with a vector bundle map $\rho : E \to T\Mfold$(called the anchor), as well as with a Lie bracket $[ \mbox{ },\mbox{ } ]_E$ on the space 
 $C^\infty(M,E)$ of smooth sections of $E$, satisfying
 \[
 \rho \circ [X,Y]_E =[\rho \circ X,\rho \circ Y],
 \]
where the right-hand side is the usual commutator of vector fields on $C^\infty(\Mfold,T\Mfold)$,
and 
\[ [X,fY]_E =f[X,Y]_E +((\rho \circ X)f)Y\]
 for all $X,Y \in C^\infty(M,E)$ and $f \in C^\infty(\Mfold)$.
 \end{definition}

\begin{remark}
A Lie algebroid is also a Lie groupoid with groupoid product given by fibre-wise addition.
\end{remark}
 \begin{example}\mbox{ }
 \label{algebroidex}
 \begin{enumerate}

\item
A Lie algebra $\g$ with its Lie bracket is a Lie algebroid over a point.
 
 \item
The tangent bundle $T\Mfold$ of a manifold $\Mfold$  defines a Lie algebroid under the Lie bracket of vector fields, and the anchor map $\rho : T\Mfold \to T\Mfold$ is the identity.

  \item \label{algebroidex3}
 For a Lie groupoid $\Gg  \rightrightarrows ^{\hspace{-0.25cm}^{\RA}}
_{\hspace{-0.25cm}_{\SO}}\Mfold$, let
 $A(\Gg)$ be the normal vector bundle defined by the embedding $\Mfold   \hookrightarrow \Gg$, with
 bundle projection given by $s$. Identify the normal bundle by $\KER d r \rvert _\Mfold $, the anchor map is given by
 $\rho:=ds: \KER dr\to T\Mfold$.
 Finally, by identifying $\KER d r$ with 
$C^\infty(\Gg, T \Gg)^L$, equip $A(\Gg)$ with the Lie bracket coming from $C^\infty(\Gg, T \Gg)$.
 $A(\Gg)$ is a Lie algebroid.
  
 \item
 Following the construction in the first example. $T\Mfold \times \g$ is the Lie algebroid of the Lie groupoid $M\times M \times G$. The anchor map $\rho$ consist of bundle projection on $T\Mfold$ and zero on $\g$.
 
 \end{enumerate}
 \end{example}
 
 In Example~\ref{algebroidex}\eqref{algebroidex3} above, $A(\Gg)$ is called the associated Lie algebroid of the Lie groupoid $\Gg  \rightrightarrows ^{\hspace{-0.25cm}^{\RA}}
_{\hspace{-0.25cm}_{\SO}}\Mfold$.
It is in itself a groupoid over $M$, similar to $TM$.
\begin{theorem}[\cite{mackenzie}]
Given the associated Lie algebroid of $A(\Gg)$ of the Lie groupoid $\Gg  \rightrightarrows ^{\hspace{-0.25cm}^{\RA}}
_{\hspace{-0.25cm}_{\SO}}\Mfold$. There exists a unique local diffeomorphism 
$\EXP:
A(\Gg) \to \Gg$.
\end{theorem}
 
\begin{definition}
Let $\Gg  \rightrightarrows ^{\hspace{-0.25cm}^{\RA}}
_{\hspace{-0.25cm}_{\SO}}\Mfold$ be a Lie groupoid with Lie algebroid $A(\Gg)$.
The {\bf tangent groupoid} $\Tt \Gg$ of  $\Gg $ is the Lie groupoid $\Tt\Gg$ over the base
$M\times [0,1]$, such that
\begin{itemize}
\item As a set, $\Tt \Gg:=  A(\Gg) \times \{0\}  \bigsqcup \Gg \times (0,1] $;
\item
And $A(\Gg)$ and $\Gg$ are glued together by the local diffeomorphism
$\EXP: A(\Gg) \to \Gg$.
\end{itemize}

\end{definition}
We do not elaborate the definition of $\EXP: A(\Gg) \to \Gg$ here, but rather attempt to illustrate it with an example below.

\begin{example}[\cite{ncg}]\mbox{}
\label{tangentgroupoid}
\begin{enumerate}
\item
The tangent groupoid $\Tt G$ of a Lie group $G$ is just $\g  \times \{0\} \bigsqcup G \times (0,1]$ glued 
together with the exponential map.
\item
The  tangent groupoid $\Tt \Mfold$ of a manifold $\Mfold$ is 
$\Tt \Mfold= T\Mfold \times \{ 0 \} \bigsqcup \Mfold \times \Mfold \times
(0,1]$ as a set.
And the groupoids 
 $T\Mfold $ and $ \Mfold \times \Mfold $ are glued together such that
for
$
p(0) =  
(x,V_x,0) \in T\Mfold \times \{0\}$,
then $p(\hbar)= \left(x,\exp \hbar V_x, \hbar\right)\in \Mfold \times \Mfold \times (0,1]$.
\end{enumerate}
\end{example}

We think of an element $(x,y,\hbar)$ of $\Mfold\times \Mfold \times (0,1]$ as
a geodesic starting at $x$ and ending at $y$, and $\hbar$ is the time it takes to travel from $x$ to $y$ 
in a given velocity. Hence, it is considered a one dimensional object.

\section{Connections as Functions on Tangent Groupoid}
\label{sec4}
In this section, we propose using a smooth groupoid to probe the connection space with, and the holonomies will be encoded by the smooth $G$-valued functions over the smooth groupoid.

Denote by $\Aa_\hbar$ the space of smooth functions from $\Mfold \times \Mfold$ to $G$ and 
$\Aa_0$ the space of smooth functions from $T\Mfold$ to $\g$ such that $A_0(x,\cdot):
T_x \Mfold \to \g$ is linear for $A_0 \in \Aa_0$ and each $x\in \Mfold $.\\

Here we think of an element $A_\hbar$ of $\Aa_\hbar$ as a holonomy presentation of a connection
for paths described by $\Mfold\times \Mfold \times (0,1]$.

\begin{definition}
\label{qconnection}
Define the space of {\bf q-connections} $\Func(\Tt \Mfold, \Tt G)$ to be
$\Aa_0 \times \{0\} \bigsqcup \Aa_\hbar \times (0,1]$ as a set.
And $\Aa_0$ and $\Aa_\hbar$ are glued together as
\[A_\hbar (x, \exp \hbar V_x ) = \exp (\hbar A_0(x,V_x))\] for all $\hbar \in [0,1]$.
\end{definition}

The definition is inspired by the following intuition.
The two points $x$ and $\exp \hbar V_x$ are connected by the geodesic $\gamma:t\to \exp t V_x$ for $t\in [0,\hbar]$.
And the holonomy along $\gamma$ is some group element. As $\hbar$ approaches to zero and
 the end point of $\gamma$ shrinks to its starting point $x$, the holonomy contribution gets closer to the identity element in the group. Therefore, the infinitesimal of the geodesic $\gamma$ at $x$ gives the infinitesimal change in the group; so for a point $(x,V_x)$ in $T\Mfold$, one associates it an element in $\g$.

One observes the following remarks.
\begin{remark}
The gluing condition of $\Func(\Tt \Mfold, \Tt G)$ implies that \[A_h (x,y) \cdot A_h (y,x) \longrightarrow 
^{\hspace{-0.6cm}\hbar \to 0 } I_G\] and \[A_h(x,x) \longrightarrow 
^{\hspace{-0.6cm}\hbar \to 0 } I_G,\] where $I_G$ is the identity of $G$.
\end{remark}

\begin{remark}
Each $A_0\in \Aa_0$ gives rise to a $\g$-valued 1-form in a unique way. Hence,
$\Aa_0$ is naturally identified as the space of $G$-connections $\Aa$.
As a result, there is an embedding \[\Aa  \hookrightarrow \Func(\Tt \Mfold, \Tt G).\]
\end{remark}

$\Aa_\hbar$ forms a group under point-wise multiplication of $G$, and 
$\Aa_0$ forms a group under point-wise addition of $\g$.

\begin{proposition}
The product $\Func(\Tt \Mfold, \Tt G)$  inherits  from $\Aa_0$ and $\Aa_\hbar$ is smooth.
\end{proposition}
\begin{proof}
The proof follows from
\[\left. \frac{d}{d\hbar}(A_\hbar \cdot A'_\hbar) ( x, \exp \hbar V_x) \right\rvert_{\hbar=0}  =  (A_0+A'_0 )(x,V_x) .\]
\end{proof}


The q-connection space $\Func(\Tt \Mfold, \Tt G)$ is a package that
captures information about probing the $G$-connection space with the tangent groupoid.
That is, an element $A_\hbar$ of $\Aa_\hbar$ evaluated at $(x,y)\in \Mfold \times \Mfold$
gives the holonomy of a connection along the geodesic from $x$ to $y$.
This holonomy presentation has a natural gauge action given by
applying a symmetry at the starting point $x$, then evaluate the holonomy along the geodesic 
from $x$ to $y$, and finally apply a reverse symmetry at $y$.
We make it formal by the following definition.

Denote by  $C^\infty(\Mfold,G)$ the set of smooth functions from $\Mfold$ to $G$.
Define the {\bf gauge action} of $C^\infty(\Mfold,G)$  on $\Aa_\hbar$ by
\begin{equation}
\label{eqn:qgaugeact}
( g\cdot A_\hbar)(x,y) := g(x) A_\hbar(x,y) g^{-1}(y)
\end{equation}
for $g\in C^\infty(\Mfold,G)$ and $A_\hbar \in \Aa_\hbar$.
And define the gauge action of  $C^\infty(\Mfold,G)$ on $\Aa_0$ by
\begin{equation}
\label{eqn:cgaugeact}
(g\cdot A_0)(x,V_x):= g(x) A_0(x,V_x) g^{-1}(x) + g(x) \langle dg^{-1}(x),V_x\rangle
\end{equation}
for $g\in C^\infty(\Mfold,G)$ and $A_0 \in \Aa_0$.
\begin{remark}
Equation~\eqref{eqn:cgaugeact} is the usual gauge action on $G$-connections.
\end{remark}

The following proposition shows that the gauge actions defined on $\Aa_\hbar $ and $\Aa_0$ are compatible.
\begin{proposition}
The  
$C^\infty(\Mfold,G)$ action  on 
the q-connection space $\Func(\Tt \Mfold, \Tt G)$ induced from Equations~\eqref{eqn:qgaugeact},\eqref{eqn:cgaugeact} is smooth.
\end{proposition}
\begin{proof}
Suppose that $A_\hbar (x,\exp \hbar V_x) = \exp \hbar A_0 (x,V_x)$.
Then 
\begin{eqnarray*}
\left. \frac{d}{d\hbar} (g\cdot A_\hbar)(x,\exp\hbar V_x) \right\rvert _{\hbar=0}
&=&
 \left. \frac{d}{d\hbar} 
 \left( g(x) A_\hbar (x,\exp \hbar V_x) g^{-1} (\exp \hbar V_x)\right) \right\rvert _{\hbar=0}\\
 &=&
 \left. g(x)\left( \frac{d }{d\hbar} A_\hbar (x,\exp \hbar V_x) \right)g^{-1} (\exp \hbar V_x) \right\rvert _{\hbar=0}
 \\&&+
 \left. g(x) A_\hbar(x, \exp \hbar V_x )\left( \frac{d }{d\hbar}g^{-1}(\exp \hbar V_x)  \right)\right\rvert_{\hbar=0}\\
 &=&
 g(x) A_0 (x,V_x)g^{-1}(x) + g(x)\langle d g^{-1} (x), V_x\rangle\\
 &=& (g\cdot A_0 )(x, V_x) .
\end{eqnarray*}
The proof is complete.
\end{proof}

Let $\DIFF(\Mfold)$ denote the diffeomorphism group of $\Mfold$.
 $\Tt \Mfold=T\Mfold \times \{ 0 \} \bigsqcup \Mfold \times \Mfold \times
(0,1]$ carries a smooth $\DIFF(\Mfold)$ action given by
\[
\begin{array}{rl}
(x,y)\mapsto \left(\sigma(x),\sigma(y) \right) & \mbox{ for } (x,y) \in \Mfold \times \Mfold 
\mbox{ and } \sigma \in \DIFF(\Mfold) ,\\
(x,V_x)\mapsto  (\sigma(x), d\sigma_{\sigma(x)} V_x ) & \mbox{ for } (x,V_x) \in T\Mfold 
\mbox{ and } \sigma \in \DIFF(\Mfold) .
\end{array}
\]
Subsequently, $\DIFF(\Mfold)$ acts on the q-connection space $\Func(\Tt \Mfold, \Tt G)$ smoothly via the induced action
\[
(\sigma \cdot A_\hbar ) (x,y) = A_\hbar  \left(\sigma(x), \sigma(y) \right)
\mbox{ and }
(\sigma \cdot A_0 )(x, V_x) = A_0 \left(\sigma(x), d\sigma_{\sigma(x)} V_x\right),
\]
for $\Mfold \in \DIFF(\Mfold)$.

\section{Connections as Operators on Hilbert Space}
\label{sec5}
Suppose that $G $ unitarily represents on some finite dimensional vector space, say without loss
of generality $\Cc^N$. Then $G$
 is included in the  matrix algebra $M_N(\Cc)$ as unitary matrices.
 Let us fix  an
 orientation on $\Mfold$, hence a volume form.
 Then one obtains the Hilbert space 
 $L^2(\Mfold, \Cc^N)$ that $\Aa_\hbar$ acts on by convolution
 \[
 (A_\hbar * \varphi)(y)= \int _\Mfold A_\hbar (x,y) \cdot \varphi (x) dx ,
 \]
 where $A_\hbar \in \Aa_\hbar$, $\varphi \in L^2(\Mfold,\Cc^N)$, the dot $\cdot$ is the
 unitary representation of $G$.
 
Denote by ${\Aa_\hbar^\#}$ the space of $M_N(\Cc)$-valued smooth functions on $\Mfold \times \Mfold$, thus 
${\Aa_\hbar^\#} \supset \Aa_\hbar$ and it acts on $L^2(\Mfold,\Cc^N)$.
Elements of ${\Aa_\hbar^\#}$ will again be denoted by $A_\hbar$.
${\Aa_\hbar^\#}$ comes equipped with an involution  given by the point-wise conjugate transpose of $M_N(\Cc)$.
The action of ${\Aa_\hbar^\#}$ on $L^2(\Mfold,\Cc^N)$ gives rise to 
a noncommutative product on
${\Aa_\hbar^\#} $
given by the convolution 
\[
(A_\hbar * A'_\hbar)(x, z)= \int _\Mfold A_\hbar ( x,y) \cdot A'_\hbar (y,z) dy ,
\]
for $A_\hbar, A'_\hbar \in {\Aa_\hbar^\#}$.

The space ${\Aa_\hbar^\#}$ forms a $*$-algebra, 
and it is identified with the ideal of trace-class operators 
 on the Hilbert space $L^2(\Mfold, \Cc^N)$.

Let $\TR$ denote the operator trace on $L^2(\Mfold,\Cc^N)$,
 and $\TR:{\Aa_\hbar^\#} \to \Cc$ is explicitly given by
\[
\TR (A_\hbar) = \int _M \Tr A_\hbar (x,x) dx ,
\]
where $\Tr$ is the matrix trace of $M_N(\Cc)$.

The linear functional $\TR: {\Aa_\hbar^\#} \to \Cc$ is invariant under the gauge action of $C^\infty(\Mfold,G)$,
as
\begin{eqnarray*}
\TR ( g \cdot A_\hbar) &=& \int _M \Tr \left( g(x) A_\hbar(x,x) g^{-1}(x) \right) dx \\
&=&  \int _M \Tr \left( A_\hbar(x,x) \right) dx \\
&=& \TR(A_\hbar) .
\end{eqnarray*}
Such a property is called {\bf gauge invariant}. 
The group element $A_\hbar(x,x)$ represents the holonomy of a connection around a loop with
base point $x$. The functional $\TR:{\Aa_\hbar^\#} \to \Cc$ being gauge invariant
is parallel to the fact that loop variables being gauge invariant in loop quantum gravity.

The package of 
${\Aa_\hbar^\#}$ acting on  $L^2(\Mfold,\Cc^N)$ with gauge group $C^\infty(M,G)$ resembles the noncommutative standard model, where the algebra is given by the the gauge group, the Hilbert space is unchanged, and the resolvent of the Dirac operator gives rise to an element of ${\Aa_\hbar^\#}$.

\section{Strict Deformation Quantisation of $G$-Connections}
\label{sec6}
\begin{definition}
A {\bf system of  Haar measures} for a groupoid $\Gg  \rightrightarrows ^{\hspace{-0.25cm}^{\RA}}
_{\hspace{-0.25cm}_{\SO}}\Mfold$ is a family of measures $(\lambda^x)_{x\in \Mfold}$, 
where each $\lambda^x$ is a positive, regular, Borel measure on $\Gg^x = s^{-1}(x)$.
\end{definition}

\begin{theorem}[\cite{landsman}]
For every Lie groupoid $\Gg  \rightrightarrows ^{\hspace{-0.25cm}^{\RA}}
_{\hspace{-0.25cm}_{\SO}}\Mfold$, there exists
a smooth system of Haar measures $(\lambda^x)_{x\in \Mfold}$, so that the convolution product 
with respect to the Haar system,
\[
( \varphi \ast \psi )(\gamma) := \int_{\Gg^{s(\gamma)}} f(\gamma  \gamma_1^{-1}) g(\gamma_1) d \lambda^{s(\gamma)} (\gamma_1) 
\mbox{ for }
\varphi, \psi \in C(\Gg),
\]
together with the star structure,
\[
\varphi^*(\gamma):=\overline{\varphi(\gamma^{-1})},
\]
 give rise to a $C^*$-algebra structure on the space of continuous functions
$C^*(\Gg)$ on $\Gg$.
\end{theorem}
Here we are not being specific on the $C^*$-norm and its closure. For our case of a Lie groupoid, the different closures coincide.

\begin{example}[\cite{landsman}]\mbox{}
\begin{enumerate}
\item
$C^*(G)$ of a Lie group $G$ is the  group $C^*$-algebra.
\item
Given a smooth system of Haar measures $(\lambda^x)_{x\in \Mfold}$, where each $\lambda ^x$ is a Haar measure on $T_x \Mfold$, and  two smooth functions $\varphi,\psi$ on $T\Mfold$. The product
\[
(\varphi\ast \psi ) (V_x):=  \int _{W_x \in T_x \Mfold} \varphi(V_x - W_x) \psi(W_x) d\lambda ^x
\]
defines, by continuity, a $C^*$-algebra structure on the space of continuous functions on $T\Mfold$, denoted
$C^*(T\Mfold)$.
\item
Given a smooth system of Haar measures $(\lambda^x)_{x\in \Mfold}$, where each $\lambda ^x$ is a Haar measure on $M$, and two smooth functions $\varphi', \psi'$ on $M\times M$. The product
\[
(\varphi\ast \psi ) (x,z):=  \int _{y \in \Mfold} \varphi'(x,y) \psi'(y,z) d\lambda ^x
\]
defines, by continuity, a $C^*$-algebra structure on the space of continuous functions on $\Mfold \times \Mfold$, denoted
$C^*(\Mfold \times \Mfold)$.
\end{enumerate}
\end{example}

Similarly, there is associated a groupoid $C^*$-algebra $C^*(\Tt \Mfold)$ to $\Tt \Mfold$, which  can also be seen 
as gluing the groupoid $C^*$-algebras $C^*(T\Mfold)$ and $C^*(M\times M)$ together.

By Fourier transform, the groupoid $C^*$-algebra $C^*(T\Mfold)$ of $T\Mfold$ is isomorphism to 
the continuous function algebra $C_0(T^*\Mfold)$ on the cotangent bundle under point-wise multiplication. 
This algebra contains the Poisson algebra $C^\infty(T^*\Mfold)$. Therefore, one has a Poisson structure on (a sub-algebra of)
$C^*(T\Mfold)$. Similarly,  one has a Poisson structure on (a sub-algebra of) $C^*(\g)$.
\begin{definition}
A {\bf continuous field of $C^*$-algebras $\left(C, \{\Bb_\hbar,\varphi_\hbar\}_{\hbar \in [0,1] }\right)$} over $[0,1]$ consists of a $C^*$-algebra $C$, 
$C^*$-algebras $\Bb_\hbar$, $\hbar \in [0,1]$, with surjective homomorphisms $\varphi_\hbar: C \to \Bb_\hbar$ and an action of $C([0,1])$ on $C$ such that
for all $c\in C$:
\begin{enumerate}
\item
the function $\hbar \mapsto \lVert \varphi_\hbar (c)\rVert$ is continuous;
\item
$\lVert c \rVert = \sup _{\hbar \in [0,1] }  \lVert \varphi_\hbar(c) \rVert$;
\item
for $f\in C([0,1])$, $\varphi_\hbar (f c ) = f(\hbar) \varphi_\hbar (c)$.
\end{enumerate}
\end{definition}

\begin{example}[\cite{landsman}]
\label{fieldalgebra}
For the tangent groupoid $\Tt \Gg$, we define $\Tt \Gg_\hbar := \Gg \times \{\hbar\}$ for $\hbar \neq 0$ and $\Tt\Gg_0 = A(\Gg)$.
The pullback of the inclusion $\Tt\Gg_\hbar \hookrightarrow \Tt\Gg$ induces a map 
$\varphi_\hbar : C^\infty_c (\Tt\Gg) \to C^\infty_c (\Tt\Gg_\hbar)$, which extends by continuity to a surjective $\ast$-homomorphism 
$\varphi_\hbar : C^\ast (\Tt\Gg) \to C^\ast (\Tt\Gg_\hbar)$.
The $C^*$-algebras $C:=C^*(\Tt\Gg)$ and $\Bb_\hbar := C^*(\Tt\Gg_\hbar)$ with the maps $\varphi_\hbar$ for a continuous field over $[0,1]$.
\end{example}

\begin{definition}[\cite{rieffel}]
A strict deformation quantisation of a Poisson manifold $P$ consists of
\begin{enumerate}
\item a continuous field of $C^*$-algebras $(C, \{\Bb_\hbar,\varphi_\hbar \}_{\hbar \in [0,1]})$, with $\Bb_0 = C_0(P)$;
\item a dense Poisson algebra $\Bb^0 \subset C_0(P)$ under the given Poisson bracket 
$\{\mbox{ } ,\mbox{ } \}$ on $P$;
\item a linear map $\Qq:\Bb^0 \to C$  that satisfies (with $\Qq_\hbar (f) := \varphi_\hbar (\Qq (f) )$ )
\begin{eqnarray*}
\Qq_0(f)&=&f , \\
\Qq_\hbar (f^*) &=& \Qq_\hbar (f)^*,
\end{eqnarray*}
for all $f\in \Bb^0$ and $\hbar \in I$, and for all $f,g \in \Bb^0$ satisfies the Dirac condition
\[
\lim _{\hbar \to 0}\left\lVert (i \hbar )^{-1} [\Qq_\hbar(f), \Qq_\hbar (g)] - \Qq_\hbar ( \{f,g\} ) \right\rVert = 0 .
\]
\end{enumerate}
\end{definition}

\begin{theorem}[\cite{landsman}]
\label{liequantisation}
Let $\Gg$ be a Lie groupoid and $A(\Gg)$ its associated Lie algebroid.
The continuous field of $C^*$-algebras $\left(  C^*(\Tt\Gg), \{C^*(\Tt\Gg_\hbar), \varphi_\hbar \}_{\hbar \in [0,1]}\right)$, as defined in Example~\ref{fieldalgebra},
defines a strict deformation quantisation of the Poisson manifold $A^*(\Gg)$.
\end{theorem}

The Poisson structure on $A^*(\Gg)$ is induced dually by the Lie bracket on $A(\Gg)$.

\begin{example}\mbox{}
\begin{enumerate}
\item
$C^\infty(\g^*)$ is a Poisson algebra with Poisson bracket induced from the Lie bracket of $\g$.
Take $C$ to be the $C^*$-algebra generated by the tangent groupoid $\Tt(G)$, $\Bb_0$ to be $C^*(\g)$, and
$\Bb_\hbar$ to be the group $C^*$-algebra $C^*(G)$ for $\hbar \neq0$.
The quantisation map $Q$ is the inclusion of $C^\infty(\g^*)$ into $C^*(\Tt G)$, and $Q_\hbar$
is $Q$ followed by map induced by the inclusion of $\g$ or $G$ into $\Tt (G)$.
\item
$T^*\Mfold$ is a symplectic manifold. Its symplectic structure defines the Poisson algebra $C^\infty(T^*\Mfold)$, which includes into continuous function algebra $C^0(T^*\Mfold)$. By Fourier transform, $C^0(T^*\Mfold)$ is isomorphic
to $C^*(T\Mfold)$. The groupoid $T\Mfold$ exponentiates to $\Mfold\times \Mfold$, thus one obtains
the inclusion of $C^\infty(T^*\Mfold)$ into the $C^*$-algebra $C^*(\Tt \Mfold)$.
\item
Take the Lie groupoid $M\times M \times G$, it has the associated Lie algebroid $T\Mfold \times \g$.
$C^\infty(T^*\Mfold \times \g^*)$ is a dense Poisson algebra in $C^*(T\Mfold \times \g)$.
The continuous field of $C^*$-algebras, $(C, \{\Bb_\hbar, \varphi_\hbar \}_{\hbar \in \Rr})$
is given by $C= C^*(\Tt \Mfold)$, $\Bb_\hbar=C^*(M\times M)$ for $\hbar \neq 0$, 
and $\Bb_0=C_0(T^*\Mfold\times \g^*)=C^*(T\Mfold \times \g)$. The quantisation map 
$Q: C^\infty ( T^*\Mfold \times \g^*) \to C^*(\Tt \Mfold)$ is the inclusion, and
 $Q_\hbar : C^\infty (T^*\Mfold \times \g^* ) \to  \Bb_\hbar$ is $Q$ followed by the restriction map.
\end{enumerate}
\end{example}

Recall that  a $q$-connection $A_\hbar \in \Func(\Tt \Mfold, \Tt G)$ is precisely a $\g$-valued 1-form when $\hbar=0$.
By considering (the characteristic function supported 
on) the graph of the function $A_\hbar$, one obtains  a distribution on
the tangent groupoid
 $\Tt (\Mfold \times \Mfold \times G)$ of the Lie groupoid $\Mfold \times \Mfold \times G$.
Therefore, one has an action of the space of connections $\Aa$ on the $C^*$-algebra $C^*(\Tt\Gg)$, for $\Gg=M\times M \times G$.
 By a smearing the characteristic function, that is integrating the distribution	 with
 some smooth function,
 one turns the characteristic function into an element in $C^*(\Tt\Gg)$.
 Therefore, a map from from  $\Func(\Tt \Mfold, \Tt G)$, which includes the space of ordinary connections $\Aa$, to $C^*(\Tt\Gg)$.
 Then Theorem~\ref{liequantisation} allows us to strictly deformation quantise the connections.
 Therefore, what one has obtained in this procedure is a strict deformation quantising of the space of connections $\Aa$, which is obtained by mapping $\Aa$ into a 
 $C^*$-algebra that constitutes the properties of a strict deformation quantisation.
 The goal of this formalism is to provide a deforming parameter $\hbar$, so that when $\hbar=0$, ordinary connections are retrieved. This line of work is to provide loop quantum gravity a semi-classical limit, so that the theory of classical gravity, general relativity, returns as one takes the limit $\hbar \to 0$.

\section{Outlook}
\label{sec7}
The lack of a semi-classical limit in loop quantum gravity has been a long standing problem. The main obstacle of obtaining such a limit is by nature how one traditionally constructs a measure on the space of connections -- probing the connection space with a collection of finite graphs as described in Section~\ref{sec2}. Finite graphs are very rigid and discrete, they do not provide a parameter that one usually encounters in quantum theory to adjust. As a result, obtaining a semi-classical limit simply becomes impossible if a ``smoother'' way of probing the connection space is not introduced. To come up with a smooth way of probing the connection space, we understand this probing procedure as evaluating some one-dimensional objects, which is a groupoid. Thus, by using a smooth groupoid, or more specifically a Lie groupoid, one has a hope of circumventing the discreteness problem. The proposal we give here is the tangent groupoid. As it turns out, when the right tangent groupoid is used, the space of connections includes into the groupoid. From there, one could deform the connections to convolution operators acting on a Hilbert space. This formalism recreates some elements appear in the noncommutative standard model \cite{stdtriple} in the way that, the gauge group is the unitary part of the algebra in the spectral triple of noncommutative standard model, the Hilbert space remains the same, and the connections are realized as trace-class operators, with the trace being a gauge invariant quantity that resembles the loop variables in loop quantum gravity literature.
The tangent groupoid provides another important feature concerning deformation, which is the strict deformation quantisation result of Landsman \cite{landsman}. Deformation quantisation can naturally be formulated in terms of $C^*$-algebraic language, a theorem due to Landsman shows that a tangent groupoid gives rise to a deformation quantisation in the strict sense. By combining Landsman's result and our realization of connections using tangent groupoids, we obtain a deformation quantisation of $G$-connections. And this quantisation formalism permits the existence of a parameter $\hbar$, so that when $\hbar=0$, one retrieves the classical connections.
However, this is not the end of the story. 
Since connection variables are only half of the variables in gravity in Arnowitt-Deser-Misner formulation, one still has to look into the other half of the variables that the connections are conjugate dual to, tetrads or metrics. It is known that tetrads are quantised to degree one differential operators \cite{AGN,lai}, and its semi-classical limit can be obtained from the $\hbar \to 0$ limit of integral kernel of the differential operator multiplied by $\hbar$, which gives nothing but the symbol of the differential operator \cite{lai2}. In the case of the manifold $\Mfold$ being three dimensional, the symbol is an $\su(2)$-valued function 
the Poisson manifold $T^*\Mfold$. Hence the symbol of the differential operator lives in the same space as a connection does. Following up the work of deformation quantisation of connections here, the next step is to examine the interaction of the tetrads with the connections at both the classical level $\hbar=0$ and the quantum level $\hbar \neq 0$. At $\hbar = 0$, one needs to examine the Poisson bracket of a connection and a symbol (of a differential operator), and determine which symbol is conjugate dual to a given connection. At $\hbar \neq 0$, one repeats the same procedure except now that the Poisson bracket is replaced with a commutator. We will leave those considerations to another article. Finally, the author would like to stress that the line of work here is toward obtaining a semi-classical limit in loop quantum gravity, and this tangent groupoid application has not been seriously considered before, thus the work here may appear incomplete. Nonetheless, the development so far shows that a semi-classical limit in loop quantum gravity is more within reach than before.

\bibliographystyle{amsalpha}

\end{document}